\documentclass[11pt]{article}

\usepackage{microtype}
\usepackage{algorithm,algorithmic}
\usepackage{cleveref,xspace}
\usepackage{bbm}
\usepackage{times}
\usepackage{graphicx}
\usepackage{amsmath,amssymb,amsthm,mathtools}
\usepackage{paralist}
\usepackage{bm}
\usepackage{xspace}
\usepackage{url}
\usepackage{fullpage, prettyref}
\usepackage{boxedminipage}
\usepackage{wrapfig}
\usepackage{ifthen}
\usepackage{color}
\usepackage{xcolor}
\usepackage{framed}
\usepackage{algorithmic,algorithm}
\usepackage{fullpage}

\usepackage{thmtools}
\usepackage{thm-restate}

\newtheorem{theorem}{Theorem}[section]

\newtheorem{lemma}[theorem]{Lemma}
\newtheorem{claim}[theorem]{Claim}
\newtheorem{corollary}[theorem]{Corollary}
\newtheorem{definition}[theorem]{Definition}

\newcommand{\ignore}[1]{}

\usepackage{cleveref}

\newcommand{\nukc}{{\sf NUkC}\xspace}
\newcommand{\kcwo}{{\sf kCwO}\xspace}
\newcommand{\rmfc}{{\sf RMFC-T}\xspace}
\newcommand{\Itree}{\ensuremath{{\cal I}_T}\xspace}
\newcommand{\I}{\ensuremath{{\cal I}}\xspace}

\newcommand{\cov}{\ensuremath{{\sf Cov}}\xspace}
\newcommand{\sse}{\ensuremath{\subseteq}\xspace}

\newcommand{\cI}{{\cal I}}

\newcommand{\poly}{\mathrm{poly}}

\newcommand{\polyloglog}{\mathrm{\poly\log\log}}

\newcommand{\floor}[1]{\lfloor#1\rfloor}

\title{The Non-Uniform $k$-Center Problem}
\author{Deeparnab Chakrabarty \and Prachi Goyal \and Ravishankar Krishnaswamy}
\date{Microsoft Research, India\\{\tt dechakr,t-prgoya,rakri@microsoft.com} }

\begin{document}

\maketitle

\begin{abstract}
In this paper, we introduce and study the {\sf Non-Uniform $k$-Center} (\nukc) problem.
	Given a finite metric space $(X,d)$ and a collection of balls of radii $\{r_1\geq \cdots \ge r_k\}$, the \nukc problem is to find a placement of their centers on the metric space and find the minimum dilation $\alpha$,
	such that the union of balls of radius $\alpha\cdot r_i$ around the $i$th center covers all the points in $X$. This problem naturally arises as a min-max vehicle routing problem with fleets of different speeds.

	The \nukc problem generalizes the classic $k$-center problem when all the $k$ radii are the same (which can be assumed to be $1$ after scaling). It also generalizes the $k$-center with outliers (\kcwo for short) problem when there are $k$ balls of radius $1$ and $\ell$ balls of radius $0$. There are $2$-approximation and $3$-approximation algorithms known for these problems respectively; the former is best possible unless P=NP and the latter remains unimproved for 15 years. \smallskip
	
	We first observe that no $O(1)$-approximation is to the optimal dilation is possible unless P=NP, implying that the \nukc problem is more non-trivial than the above two problems.
	Our main algorithmic result is an $(O(1),O(1))$-{\em bi-criteria} approximation result: we give an $O(1)$-approximation to the optimal dilation, however, we may open  $\Theta(1)$ centers of each radii. Our techniques also allow us to prove a simple (uni-criteria), optimal $2$-approximation to the \kcwo problem improving upon the long-standing $3$-factor.
	
	Our main technical contribution is a connection between the \nukc problem and the so-called firefighter problems on trees which have been studied recently in the TCS community.
    We show \nukc is as hard as the firefighter problem. While we don't know if the converse is true, we are able to adapt ideas from recent works~\cite{CC10,ABZ16} in non-trivial ways to obtain our
    constant factor bi-criteria approximation.
 \end{abstract}




\section{Introduction}
Source location and vehicle routing problems are extremely well studied~\cite{Lap98,MJS98,GN11} in operations research.
Consider the following location+routing problem: we are given a set of $k$ ambulances with speeds $s_1,s_2,\ldots,s_k$ respectively, and we have
to find the depot locations for these vehicles in a metric space $(X,d)$ such that any point in the space can be served by some ambulance as fast as possible.
If all speeds were the same, then we would place the ambulances in locations $S$ such that $\max_{v\in X} d(v,S)$ is minimized -- this is the famous $k$-center problem.
Differing speeds, however, leads to non-uniformity, thus motivating the titular problem we consider.

\begin{definition}[The Non-Uniform $k$-Center Problem (\nukc)]
	The input to the problem is a metric space $(X,d)$ and a collection of $k$ balls of radii $\{r_1\ge r_2\ge\cdots \ge r_k\}$.
	The objective is to find a placement $C\subseteq X$ of the centers of these balls,
	so as to minimize the dilation parameter $\alpha$ such that the union of balls of radius $\alpha\cdot r_i$ around the $i$th center covers all of $X$.
	Equivalently, we need to find centers $\{c_1,\ldots,c_k\}$ to minimize $\max_{v\in X} \min_{i=1}^k \frac{d(v,c_i)}{r_i}$.
\end{definition}

\begin{figure}[b]
  \centering
  \includegraphics[width=\textwidth]{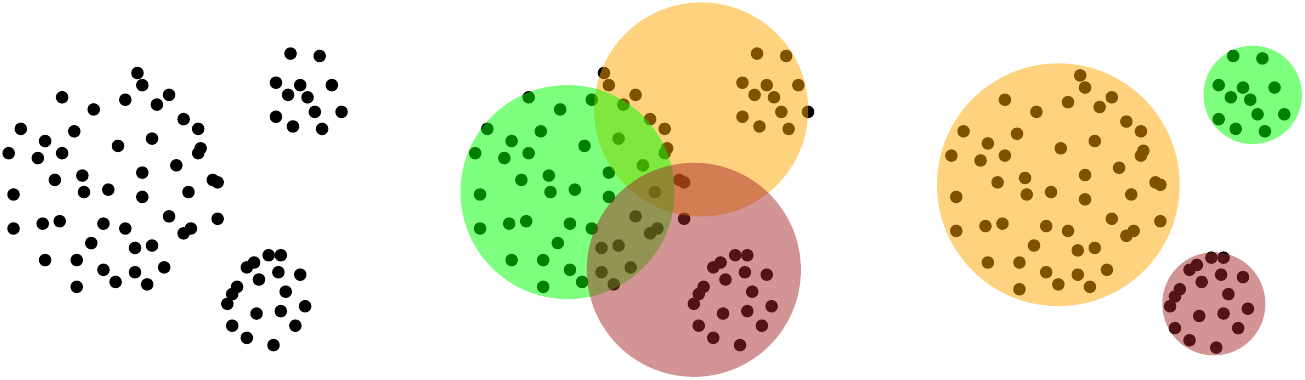}
  \caption{The left figure shows the dataset, the middle figure shows a traditional $k$-center clustering, and the right figure depicts a non-uniform clustering}\label{fig:example}
\end{figure}

As mentioned above, when all $r_i$'s are the same (and equal to $1$ by scaling), we get the $k$-center problem.
The $k$-center problem was originally studied by Gonzalez~\cite{Gon85} and Hochbaum and Shmoys~\cite{HS85} as a clustering problem of partitioning a metric space into different clusters to minimize maximum intra-cluster distances.
One issue (see Figure~\ref{fig:example} for an illustration and refer to~\cite{GRS98} for a more detailed explanation) with $k$-center (and also $k$-median/means) as an objective function for clustering is that it favors clusters of similar sizes with respect to cluster radii. However, in presence of qualitative information
on the differing cluster sizes, the non-uniform versions of the problem can arguably provide more nuanced solutions.
One such extreme special case was considered as the ``clustering with outliers'' problem~\cite{CKMN01} where some fixed number/fractions of points
in the metric space need not be covered by the clusters. In particular, Charikar et al~\cite{CKMN01} consider (among many problems) the $k$-center with outlier problem (\kcwo, for short) and show a $3$-approximation for this problem. It is easy to see that \kcwo is a special case of the \nukc problem when there are $k$ balls of radius $1$ and $\ell$ (the number of outliers) balls of radius $0$.

Motivated by the aforementioned reasons (both from facility location as well as from clustering settings), we investigate the worst-case complexity of the \nukc problem. Gonzalez~\cite{Gon85} and Hochbaum and Shmoys~\cite{HS85} give $2$-approximations for the $k$-center problem,
and also show that no better factor is possible unless P = NP. Charikar et al~\cite{CKMN01} give a $3$-approximation for the \kcwo problem, and this has been the best factor known for 15 years. Given these algorithms, it is natural to wonder if a simple $O(1)$-approximation exists for the \nukc problem.
In fact, our first result shows a qualitative distinction between \nukc and these problems: constant-approximations are impossible for \nukc unless P=NP.



\begin{theorem}\label{thm:hardness}
	For any constant $c \geq 1$, the \nukc problem does not admit a $c$-factor approximation unless $P=NP$, even when the underlying metric is a tree metric.
\end{theorem}

The hardness result is by a reduction from the so-called {\em resource minimization for fire containment}  problem on trees (\rmfc, in short), a variant of the firefighter problem. To circumvent the above hardness, we give the following bi-criteria approximation algorithm which is the main result of the paper, and which further highlights the connections with \rmfc since our algorithms heavily rely on the recent algorithms for \rmfc~\cite{CC10,ABZ16}. An $(a,b)$-factor bi-criteria algorithm for \nukc returns a solution which places at most $a$ balls of each type (thus in total it may use as many as $a\cdot k$ balls),
and the dilation is at most $b$ times the optimum dilation for the instance which places exactly one ball of each type.

\begin{theorem}\label{thm:algo}\label{thm:main-result}
	There is an $(O(1),O(1))$-factor  bi-criteria algorithm for the \nukc problem.
\end{theorem}

Furthermore, as we elucidate below, our techniques also give uni-criteria results when the number of distinct radii is $2$. In particular, we get a $2$-approximation for the \kcwo problem and a $(1+\sqrt{5})$-approximation when there are only two distinct types of radii.
\begin{theorem}
	\label{thm:kcwo}
	There is a $2$-approximation for the \kcwo problem.
\end{theorem}

\begin{theorem}\label{thm:twotypes}
There is a $(1+\sqrt{5})$-approximation for the \nukc problem when the number of distinct radii is at most $2$.
\end{theorem}

\subsection{Discussion on Techniques}

Our proofs of \Cref{thm:hardness,thm:algo} shows a strong connection between \nukc and the so-called {\em resource minimization for fire containment}  problem on trees (\rmfc, in short). This connection is one of the main findings of the paper, so
we first formally define this problem.
\begin{definition}[Resource Minimization for Fire Containment on Trees (\rmfc)]
	Given a rooted tree $T$ as input, the goal is to select a collection of non-root nodes $N$ from $T$ such that (a) every root-leaf path has at least one vertex from $N$, and (b) $\max_t |N \cap L_t|$ is minimized,
	where $L_t$ is the $t$th-{\em layer} of $T$, that is, the vertices of $T$ at exactly distance $t$ from the root.
\end{definition}
To understand the reason behind the name, consider a fire starting at the root spreading to neighboring vertices each day; the \rmfc problem minimizes the number of firefighters needed per day so as to prevent the fire spreading to the leaves of $T$.

It is NP-hard to decide if the optimum of \rmfc is 1 or not~\cite{FKMR07,KM10}. 
Given any \rmfc instance and any $c > 1$, we construct an \nukc instance on a tree metric such that in the ``yes'' case there is always a placement with dilation $=1$ which covers the metric, while in the ``no'' case even a dilation of $c$ doesn't help.
Upon understanding our hardness construction, the inquisitive reader may wonder if the reduction also works in the other direction, i.e., whether we can solve \nukc using a reduction to \rmfc problem.
Unfortunately, we do not know if this is true even for two types of radii. 
 However, as we explain below we still can use positive results for the \rmfc problem to design good algorithms for the \nukc problem.

Indeed, we start off by considering the natural LP relaxation for the \nukc problem and describe an {\em LP-aware reduction} of \nukc  to \rmfc.
More precisely, given a feasible solution to the LP-relaxation for the given \nukc instance, we describe a procedure to obtain an instance of \rmfc defined by a tree $T$, with the following properties: (i) we can exhibit a feasible fractional solution for the LP relaxation of the \rmfc instance, and (ii) given any feasible {\em integral} solution to the \rmfc instance, we can obtain a feasible integral solution to the \nukc instance which dilates the radii by at most a constant factor.
Therefore, an {\em LP-based} $\rho$-approximation to \rmfc would immediately imply $(\rho,O(1))$-bicriteria approximation algorithms for \nukc.
This already implies \Cref{thm:kcwo} and \Cref{thm:twotypes} since the corresponding \rmfc instances have no integrality gap.
Also, using a result of Chalermsook  and Chuzhoy~\cite{CC10}, we directly obtain an $\left(O(\log^* n),O(1)\right)$-bicriteria approximation algorithm for \nukc.

Here we reach a technical bottleneck: Chalermsook and Chuzhoy~\cite{CC10} also show that the integrality gap of the natural LP relaxation for \rmfc is $\Omega(\log^* n)$. When combined with our hardness reduction in~\Cref{thm:hardness} , this also implies a $(\Omega(\log^* n), c)$ integrality gap for any constant $c > 1$ for the natural LP relaxation for \nukc. That is, even if we allow a violation of $c$ in the radius dilation, there is a $\Omega(\log^* n)$-integrality gap in terms of the violation in number of balls opened of each type.

However, very recently, Adjiashvili, Baggio and Zenklusen~\cite{ABZ16} show an improved $O(1)$-approximation for the \rmfc problem.
At a very high level, the main technique in~\cite{ABZ16} is the following. Given an \rmfc instance, they carefully and efficiently ``guess'' a subset of the optimum solution, such that the natural LP-relaxation for covering the uncovered leaves has $O(1)$-integrality gap. However, this guessing procedure crucially uses the tree structure of $T$ in the \rmfc problem. Unfortunately for us though, we get the \rmfc tree only after solving the LP for \nukc, which already has an $\Omega(\log^* n)$-gap! Nevertheless, inspired by the ideas in~\cite{ABZ16}, we show that we can also efficiently preprocess an \nukc instance, ``guessing'' the positions of a certain number of balls in an optimum solution,
such that the standard LP-relaxation for covering the uncovered points indeed has $O(1)$-gap. We can then invoke the LP-aware embedding reduction to \rmfc at this juncture to solve our problem. This is quite delicate, and is the most technically involved part of the paper.

\subsection{Related Work and Open Questions}
The $k$-center problem~\cite{Gon85,HS85} and the $k$-center with outliers~\cite{CKMN01} probems are classic problems in approximation algorithms and clustering.
These problems have also been investigated under various settings such as the incremental model~\cite{CCFM97,MK08}, streaming model~\cite{COP03, MK08}, and more recently in the map-reduce model~\cite{IM15,MKCWM15}. Similarly, the $k$-median~\cite{CGST99, JV99, LS13, BPRST15} and $k$-means~\cite{JV99,KMNPSW02,HM04, KSS04} problems are also classic problems studied extensively in approximation algorithms and clustering. The generalization of $k$-median to a routing+location problem was also studied recently~\cite{GN11}. It would be interesting to explore the complexity of the non-uniform versions of these problems. Another direction would be to explore if the new non-uniform model can be useful in solving clustering problems arising in practice.


%

\section{Hardness Reduction}

In this section, we prove \Cref{thm:hardness} based on the following NP-hardness~\cite{KM10} for \rmfc.
\begin{theorem}\label{thm:np}
	\cite{KM10}
	Given a tree $T$ whose leaves are at the same distance from the root,
	it is NP-hard to distinguish between the following two cases.
	YES: There is a solution to the \rmfc instance of value $1$.
	NO: All solutions to the \rmfc instance has value $2$.
\end{theorem}

Given an \rmfc instance defined by tree $T$, we now describe the construction of our \nukc instance.
Let $h$ be the height of the tree, and let $L_t$ denote the vertices of the tree at distance exactly $t$ from the root. So, the leaves constitute $L_h$ since all leaves are at the same distance from the root.
The \nukc instance, $\cI(T)$, is defined by the metric space $(X,d)$, and a collection of balls. The points in our metric space will correspond to the leaves of the tree, i.e., $X = L_h$.
To define the metric, we assign a weight $d(e) = (2c+1)^{h-i+1}$ for each edge whose one endpoint is in $L_i$ and the other in $L_{i-1}$; we then define $d$ be the shortest-path metric on $X$ induced by this weighted tree. Finally,
we set $k = h$, and define the $k$ radii $r_1 \geq r_2 \geq \ldots \geq r_k$ iteratively as follows: define $r_k := 0$, and for $k \geq i > 1$, set $r_{i-1} := (2c+1)\cdot r_i + 2(2c+1)$. This completes the \nukc instance.
Before proceeding we make the simple observation:
	for any two leaves $u$ and $u'$ with lca $v \in L_t$, we have $d(u,u') = 2(2c+1 + (2c+1)^2 + \cdots + (2c+1)^{h-t}) = r_t$. 
The following lemma proves~\Cref{thm:hardness}.
\begin{lemma}
	If $T$ is the YES case of~\Cref{thm:np}, then $\cI(T)$ has optimum dilation $=2$.
	If $T$ is the NO case of~\ref{thm:np}, then $\cI(T)$ has optimum dilation $ \geq 2c$.
\end{lemma}
\begin{proof}
	Suppose $T$ is in the YES case, and there is a solution to \rmfc which selects at most $1$ node from each level $L_t$. If $v\in L_t$ is selected,
	then select a center $c_v$ {\em arbitrarily} from any leaf in the sub-tree rooted at $v$ and open the ball of radius $r_t$. We now need to show all points in $X = L_h$ are covered by these balls.
	Let $u$ be any leaf; there must be a vertex $v$ in some level $L_t$ in $u$'s path to the root such that a ball of radius $r_t$ is opened at $c_v$. However, $d(u,c_v) \leq d(u,v) + d(v,c_v) = 2r_t$
	and so the ball of radius $2r_t$ around $c_v$ covers $u$.
	
	Now suppose $T$ is in the NO case, and the \nukc instance has a solution with optimum dilation $ < 2c$. We build a good solution for the \rmfc instance $N$ as follows: suppose the \nukc solution opens the radius $< 2c\cdot r_t$ ball around center $u$.
	Let $v$ be the vertex on the $u$-root path appearing in level $L_t$. We then pick this node in $N$.
	Observe two things: first, this ball covers all the leaves in the sub-tree rooted at $v$ since $r_t \geq d(u,u')$ for any such $u'$. Furthermore, since the \nukc solution has only one ball of each radius, we get that $|N \cap L_t| \leq 1$. Finally, since $d(u,w) \geq 2c\cdot r_t$ for all leaves $w$ not in the sub-tree rooted at $v$,
	the ball of radius $r_t$ around $u$ doesn't contain any leaves other than those rooted at $v$. Contra-positively, since all leaves $w$ are covered in some ball, every leaf must lie in the sub-tree of some vertex picked in $N$. That is, $N$ is a solution to \rmfc with value $=1$ contradicting the NO case.
\end{proof}

\section{LP-aware reduction from \nukc to \rmfc}\label{sec:3}
For reasons which will be apparent soon, we consider instances $\cI$ of \nukc counting multiplicites.
That is, we consider an instance to be a collection of tuples $(k_1,r_1),\ldots, (k_h,r_h)$ to indicate there are $k_i$ balls of radius $r_i$.
So $\sum_{t=1}^h k_t = k$. Intuitively, the reason we do this is that if two radii $r_t$ and $r_{t+1}$ are ``close-by'' then it makes sense to round up $r_{t+1}$ to $r_t$ and increase $k_t$, losing only a constant-factor loss in the dilation.
\smallskip


\noindent
{\bf LP-relaxation for \nukc.} We now state the natural LP relaxation for a given \nukc
instance $\cI$.
For each point $p \in X$ and radius
type $r_i$, we have an indicator variable $x_{p,i} \geq 0$ for whether we place a ball of radius $r_i$ centered at $p$.
By doing a binary search on the optimal dilation and scaling, we may assume that the optimum dilation is $1$. Then, the following linear program must be feasible. Below, we use $B(q,r_i)$ to denote the set of points within distance $r_i$ from $q$.
\begin{align}
\forall p\in X, & ~~~~~~~~~~~~~\sum_{t=1}^h  ~~\sum_{q \in B(p,r_t)} x_{q,t} ~~~\geq 1 \tag{\nukc LP} \label{lp:nukc}\\
\forall t\in 1,\cdots,h & ~~~~~~~~~~~~~ \sum_{p\in X} x_{p,t} ~~~\leq k_t \nonumber
\end{align}

\noindent
{\bf LP-relaxation for \rmfc.} Since we reduce fractional \nukc to fractional \rmfc, we now state the natural LP relaxation for \rmfc on a tree $T$ of depth $h+1$. In fact, we will work with the following budgeted-version of \rmfc (which is equivalent to the original \rmfc problem --- for a proof, see~\cite{ABZ16}): Instead of minimizing the maximum number of ``firefighters'' at any level $t$ (that is $|N\cap L_t|$ where $N$ is the chosen solution), suppose we specify a budget limit of $k_t$ on $|N\cap L_t|$. The goal is the minimize the maximum dilation of these budgets. Then the following is a natural LP relaxation for the budgeted \rmfc problem on trees. Here $L = L_h$ is the set of leaves, and
$L_t$ are the layer $t$-nodes. For a leaf node $v$, let $P_v$ denote the vertex set of the unique leaf-root path excluding the root.
\begin{align}
    \nonumber & \qquad \qquad \min \alpha \\
	\forall v\in L, & ~~~~~~~~~~~~~\sum_{u\in P_v}  y_u ~~~\geq 1 \tag{\rmfc LP} \label{lp:rmfc}\\
	\forall t\in 1,\cdots, h & ~~~~~~~~~~~~~ \sum_{u\in L_t} y_{u} ~~~\leq \alpha \cdot k_t \nonumber
\end{align}

%
%
%
%

\noindent
{\bf The LP-aware Reduction to Tree metrics.}
We now describe our main reduction algorithm, which takes as input an \nukc instance $\cI = \{(X,d); (k_1,r_1),\ldots,(k_h,r_h)\}$ and a feasible solution $x$ to~\ref{lp:nukc}, and returns a budgeted \rmfc instance $\Itree$ defined by a tree $T$ along with budgets for each level, and a feasible solution $y$ to~\ref{lp:rmfc} with dilation $1$.
The tree we construct will have height $h+1$ and the budgeted \rmfc instance will have budgets precisely $k_t$ at level $1\leq t\leq h$, and the budget for the leaf level is $0$.
 For clarity, throughout this section we use the word {\em points} to denote elements of the metric space in $\cI$, and the word {\em vertices/nodes} to denote the tree nodes in the \rmfc instance that we construct.
We build the tree $T$ in a bottom-up manner, where in each round, we pick a set of far-away representative points (the distance scale increases as we move up the tree) and cluster all points to their nearest representative. This is similar to a so-called clustering step in many known algorithms for facility location (see, e.g.~\cite{CGST99}), but whereas an arbitrary set of far-away representatives would suffice in the facility location algorithms, we need to be careful in how we choose this set to make the overall algorithm work.

Formally, each vertex of the tree $T$ is mapped to some point in $X$, and we denote the mapping of the vertices at level $t$ by $\psi_t: L_t \to X$.
We will maintain that each $\psi_t$ will be injective, so $\psi_t(u) \neq \psi_t(v)$ for $u \neq v$ in $L_t$. So, $\psi^{-1}_t$ is well defined for the range of $\psi_t$.

The complete algorithm runs in rounds $h+1$ to $2$ building the tree one level per round. To begin with, the $\psi_{h+1}$ mapping is an arbitrary bijective mapping between $L := L_{h+1}$, the set of leaves of the tree, and the points of $X$
(so, in particular, $|L| = |X|$). We may assume it to be the identity bijection. In each round $t$, the range of the mappings become progressively smaller, that is\footnote{We are using the notation $\psi(X) := \bigcup_{x\in X} \psi(x)$.}, $\psi_t(L_t) \supseteq \psi_{t-1}(L_{t-1})$. We call $\psi_{t}(L_{t})$ as the {\bf winners at level $t$}. We now describe round $t$. Let $\cov_t(p) := \sum_{q\in B(p,r_t)} x_{q,t}$ denote the fractional amount the point $p$ is covered by radius $r_t$ balls in the solution $x$.
Also define $\cov_{\geq t}(p) := \sum_{s\geq t} \cov_s(p)$ denoting the fractional amount $p$ is covered by radius $r_t$ or smaller balls. Let $\cov_{h+1}(p) = 0$ for all $p$.

\begin{algorithm}
\caption{Round $t$ of the LP-aware Reduction.}
\label{alg:kcwo}
\label{alg:createTree}
\begin{algorithmic}
\STATE {\bf Input:} Level $L_t$, subtrees below $L_t$, the mappings $\psi_s: L_s \to X$ for all $t \leq s \leq h$.
\STATE {\bf Output:} Level $L_{t-1}$, the connections between $L_{t-1}$ and $L_t$, and the mapping $\psi_{t-1}$.
\vspace{0.5ex}
\STATE Define $A = \psi_{t}(L_t)$ the set of points who are winners at level $t$.
\WHILE{$A \neq \emptyset$}
\STATE (a) Choose the point $p \in A$ with \emph{minimum coverage} $\cov_{\geq t}(p)$.
\STATE (b) Let $N(p) := \{q \in A: d(p,q) \leq 2r_{t-1}\}$ be the set of all nearby points in $A$ to $p$.
\STATE (c) Create a new tree vertex $w \in L_{t-1}$ corresponding to $p$ and set $\psi_{t-1}(w) := p$. \emph{Call $p$ a winner at level $t-1$, and each $q\in N(p)\sse A$ a loser to $p$ at this level.}
\STATE (d) Create edge $(w,v)$ for tree vertices $v \in \psi^{-1}_t(N(p))$ associated with $N(p)$ at level $t$.
\STATE (e) Set $A \leftarrow A \setminus ( N(p))$.
\STATE (f) Set $y_{w} = \cov_{t-1}(p)$.
\ENDWHILE
\end{algorithmic}
\end{algorithm}
\smallskip
\noindent
Finally, we add a root vertex and connect it to all vertices in $L_1$. This gives us the final tree $T$ and a solution $y$ which assigns
a value to all non-leaf, non-root vertices of the tree $T$. The following claim asserts well-definedness of the algorithm.
%

\begin{lemma} \label{lem:rmfcfeasible}
The solution $y$ is a feasible solution to~\ref{lp:rmfc} on $\Itree$ with dilation $1$.
\end{lemma}
\begin{proof}
The proof is via two claims for the two different set of inequalities.
\begin{claim}
For all $1 \leq t \leq h$, we have $\sum_{w \in L_t}{y_w} \leq k_t$.
\end{claim}
\begin{proof} Fix $t$. Let $W_t \subseteq X$ denote the winners at level $t$, that is, $W_t = \psi_t(L_t)$.
	By definition of the algorithm,
	$\sum_{w \in L_t} y_w = \sum_{p\in W_t} \cov_t(p)$. Now note that for any two points $p,q \in W_t$, we have $B(p,r_t) \cap B(q,r_t) = \emptyset$.
	To see this, consider the first point which enters $A$ in the $(t+1)$th round when $L_t$ was being formed. If this is $p$, then all points in the radius $2r_t$
	ball are deleted from $A$. Since the balls are disjoint, the second inequality of \nukc LP implies $\sum_{p\in W_t} \sum_{q \in B(p,r_t)} x_{q,t} \leq k_t$.
    The second summand in the LHS is precisely $\cov_t(p)$.
\end{proof}

\begin{claim} \label{cl:coverage}
For any leaf node $w \in L$, we have $\sum_{v \in P_w} y_v \geq 1$.
\end{claim}
\begin{proof}
	We start with an observation. Fix a level $t$ and a winner point $p\in W_t$. Let $u\in L_t$ such that $\psi_t(u) = p$.
	Since $W_t \subseteq W_{t+1} \subseteq \cdots \subseteq W_h$, there is a leaf $v$ in the subtree rooted at $u$ corresponding to $p$.
	Moreover, by the way we formed our tree edges in step (d), we have that $\psi_s(w') = p$ for all $w'$ in the $(u,v)$-path and hence
	$\sum_{w' \in [u,v]\textrm{-path}} y_{w'}$ is precisely $\cov_{\geq t}(p)$.
%

Now, for contradiction, suppose there is some leaf corresponding to, say point $p$, such that the root-leaf path has total $y$-assignment less than $1$. Then, pick the point, among all such unsatisfied points $p$, who appears in a winning set $W_t$ with $t$ \emph{as small as possible}.

By the preceding observation, the total $y$-assignment $p$ receives on its path from level $h$ to level $t$ is exactly $\cov_{\geq t}(p)$.
Moreover, suppose $p$ loses to $q$ at level $t-1$, i.e., $\psi^{-1}_t(p)$ is a child of $\psi^{-1}_{t-1}(q)$. In particular, it means that $q$ has also been a winner up to level $t$ and so the total $y$-assignment on $q$'s path upto level $t$ is also precisely $\cov_{\geq t}(q)$. Additionally, since $\psi^{-1}_{t-1}(q)$ became the parent node for $\psi^{-1}_t(p)$, we know that $\cov_{\geq t}(q) \leq \cov_{\geq t}(p)$ due to the way we choose winners in step (a) of the while loop. Finally, by our maximality assumption on $p$, we know that $q$ is fractionally satisfied by the $y$-solution. Therefore, there is fractional assignment of at least $(1 - \cov_{\geq t}(q))$ on $q$'s path from nodes in level $t-1$ to level $1$. Putting these observations together, we get that the total fractional assignment on $p$'s root-leaf path is at least $\cov_{\geq t}(p) + (1 - \cov_{\geq t}(q)) \geq 1$, which results in the desired contradiction.
\end{proof}
\end{proof}
The following lemma shows that any good {\em integral} solution to the \rmfc instance $\Itree$ can be converted to a good integral solution for the \nukc instance $\I$.
\begin{lemma}
\label{lemma:rmfc2kc}
Suppose there exists a feasible solution $N$ to $\Itree$ such that for all $1\leq t\leq h$, $|N\cap L_t| \leq \alpha k_t$.
Then there is a solution to the \nukc instance $\I$ that opens, for each $1\leq t\leq h$,  at most $\alpha k_t$ balls of radius $\leq 2r_{\geq t}$,
where $r_{\ge t} := r_t + r_{t+1} + \cdots + r_h$.
\end{lemma}
\begin{proof}
Construct the \nukc solution as follows: for level $1 \leq t \leq h$ and every vertex $w \in N\cap L_t$, place the center at $\psi_t(w)$ of radius $2r_{\geq t}$.
We claim that every point in $X$ is covered by some ball. Indeed, for any $p\in X$, look at the leaf $v = \psi_{h+1}(p)$, and let $w\in N$ be a node in the root-leaf path.
Let $w \in L_t$ for some $t$. Now observe that $d(p,\psi_t(w)) \leq 2r_{\geq t}$; this is because for any edge $(u',v')$ in the tree where $u'$ is in $L_t$ and is the parent of $v'$, we have that
$d(\psi_{t+1}(v'),\psi_{t+1}(u')) < 2r_t$.
%
\end{proof}
\noindent
This completes the reduction, and we now prove a few results which follow easily from known results about the firefighter problem.

\begin{theorem}\label{thm:logstar}
	There is a polynomial time $(\mathcal{O}(\log^{*} n),8)$-bi-criteria algorithm for \nukc.
\end{theorem}
\begin{proof}
Given any instance $\cI$ of \nukc, we first club the radii to the nearest power of $2$ to get an instance $\cI'$ with
radii $(k_1,r_1), \cdots, (k_h,r_h)$ such that an $(a,b)$-factor solution for $\cI'$ is an $(a,2b)$-solution for $\cI$.
Now, by scaling, we assume that the optimal dilation for $\cI'$ is $1$; we let $x$ be the feasible solution to the \nukc LP.
Then, using \Cref{alg:createTree}, we can construct the tree $\Itree'$ and a feasible solution $y$ to the \rmfc LP.
We can now use the following theorem of Chalermsook and Chuzhoy~\cite{CC10}: given any feasible solution to the \rmfc LP,
we can obtain a feasible set $N$ covering all the leaves such that for all $t$, $|N\cap L_t|\leq \mathcal{O}(\log^* n)k_t$.
Finally, we can apply Lemma~\ref{lemma:rmfc2kc} to obtain a $(\mathcal{O}(\log^{*} n),4)$ solution to $\cI'$  (since $r_{\geq t} \leq 2r_t$).
\end{proof}

\begin{proof}[Proof of \Cref{thm:kcwo} and \Cref{thm:twotypes}]
	We use the following claim regarding the integrality gap of \rmfc LP for depth $2$ trees.
	\begin{claim}
		When $h=2$ and $k_t$'s are integers, given any fractional solution to \rmfc LP, we can find a feasible integral solution as well.
	\end{claim}
	
	\begin{proof}
		Given a feasible solution $y$ to~\ref{lp:rmfc}, we need to find a set $N$ such that $|N\cap L_t| \leq k_t$ for $t=1,2$.
		There must exist at least one vertex $w\in L_1$ such that $y_w \in (0,1)$ for otherwise the solution $y$ is trivially integral.
		If only one vertex $w\in L_1$ is fractional, then since $k_1$ is an integer, we can raise this $y_w$ to be an integer as well.
		So at least two vertices $w$ and $w'$ in $L_1$ are fractional.
        Now, without loss of generality, let us assume that $|C(w)| \geq |C(w')|$, where $C(w)$ is the set of children of $w$. Now for some small constant $0< \epsilon <1$, we do the following:
		$y'_w := y_{w} + \epsilon$, $y'_{w'} := y_{w'} - \epsilon$, $ \forall c \in C(w),\; y'_c := y_c - \epsilon$, and $\forall c \in C(w'),\; y'_c := y_c + \epsilon$.
				Note that $y(L_1)$ remains unchanged, $y(L_2)$ can only decrease, and root-leaf paths still add to at least $1$. We repeat this till we rule out all fractional values.
	\end{proof}
	\noindent
	To see the proof of \Cref{thm:kcwo}, note that an instance of the $k$-center with outliers problem is an \nukc instance with $(k,1),(\ell,0)$, that is, $r_1 = 1$ and $r_2 = 0$.
	We solve the LP relaxation and obtain the tree and an \rmfc solution. The above claim implies a feasible integral solution to \rmfc since $h=2$, and finally note that $r_{\geq 1} = r_1$
	for \kcwo, implying we get a $2$-factor approximation. \smallskip

\noindent
The proof of \Cref{thm:twotypes} is similar. If $r_1 < \theta r_2$ where $\theta = (\sqrt{5} + 1)/2$, then we simply run $k$-center with $k = k_1 + k_2$.
This gives a $2\theta = {\sqrt{5}+1}$-approximation. Otherwise, we apply Lemma~\ref{lemma:rmfc2kc} to get a $2(1+\frac1\theta) = {\sqrt{5}+1}$-approximation.
\end{proof}

We end this section with a general theorem, which is an improvement over Lemma~\ref{lemma:rmfc2kc} in the case when many of the radius types are close to each other, in which case $r_{\geq t}$ could be much larger than $r_t$. Indeed, the natural way to overcome this would be to group the radius types into geometrically increasing values as we did in the proof of~\Cref{thm:logstar}. However, for some technical reasons we will not be able to bucket the radius types in the following section, since we would instead be bucketing the \emph{number of balls} of each radius type in a geometric manner. Instead, we can easily modify~\Cref{alg:createTree} to build the tree by focusing only on radius types where the radii grow geometrically.

\begin{theorem}
\label{thm:general-redux}
Given an \nukc instance $\I = \{M = (X,d), (k_1, r_1), (k_2, r_2), \ldots, (k_h, r_h)\}$ and an LP solution $x$ for~\ref{lp:nukc}, there is an efficient reduction which generates an \rmfc instance $\Itree$ and an LP solution $y$ to~\ref{lp:rmfc}, such that the following holds:
\begin{enumerate}
\item[(i)] For any two tree vertices $w \in L_t$ and $v \in L_{t'}$ where $w$ is an ancestor of $v$ (which means $t \leq t'$), suppose $p$ and $q$ are the corresponding points in the metric space, i.e., $p = \psi_{t}(w)$ and $q = \psi_{t'}(v)$, then it holds that $d(p,q) \leq 8 \cdot r_t$.
\item[(ii)] Suppose there exists a feasible solution $N$ to $\Itree$ such that for all $1\leq t\leq h$, $|N\cap L_t| \leq \alpha k_t$.
Then there is a solution to the \nukc instance $\I$ that opens, for each $1\leq t\leq h$,  at most $\alpha k_t$ balls of radius at most $8 \cdot r_t$.
\end{enumerate}
\end{theorem}

\subsection{Proof of~\Cref{thm:general-redux}}
Both the algorithm as well as the proof are very similar, and we now provide them for completeness. At a high level, the only difference occurs when we identify and propagate winners: instead of doing it for each radius type, we identify barrier levels where the radius doubles, and perform the clustering step only at the barrier levels. We now present the algorithm, which again proceeds in rounds $h+1,h, h-1, \ldots, 2$, but makes jumps whenever there are many clusters of similar radius type. To start with, define $r_{h+1} = 0$.


\begin{algorithm}
\caption{Round $t$ of the Improved Reduction.}
\label{alg:kcwo-new}
\label{alg:createTree-new}
\begin{algorithmic}
\STATE {\bf Input:} Level $L_t$, subtrees below $L_t$, the mappings $\psi_s: L_s \to X$ for all $t \leq s \leq h$.
\STATE {\bf Output:} Level $L_{t-1}$, the connections between $L_{t-1}$ and $L_t$, and the mapping $\psi_{t-1}$.
\vspace{0.5ex}
\STATE Let $t' = \min_s {\rm s.t. } r_s \leq 2 r_{t-1}$ be the type of the largest radius smaller than $2 r_{t-1}$.
\STATE Define $A = \psi_{t}(L_t)$ the set of points who are winners at level $t$.
\WHILE{$A \neq \emptyset$}
\STATE (a) Choose the point $p \in A$ with minimum coverage $\cov_{\geq t}(p)$.
\STATE (b) Let $N(p) := \{q \in A: d(p,q) \leq 2r_{t'}\}$ denote all points in $A$ within $2 r_{t'}$ from $p$.
\STATE (c) Create new vertices $w_{t-1}, \ldots, w_{t'-1} \in L_{t-1}, \ldots, L_{t'-1}$ levels respectively, all corresponding to $p$, i.e., set $\psi_{i}(w) := p$ for all $t'-1 \leq i \leq t-1$. Connect each pair of these vertices in successive levels with edges. Call $p$ a {\bf winner} at levels $t-1, \ldots, t'-1$.
\STATE (d) Create edge $(w_{t-1},v)$ for vertices $v \in \psi^{-1}_t(N(p))$ associated with $N(p)$ at level $t$.
\STATE (f) Set $A \leftarrow A \setminus ( N(p))$.
\STATE (g) Set $y_{w_i} = \cov_{i}(p)$ for all $t-1 \leq i \leq t'-1$.
\ENDWHILE
\STATE Jump to round $t'-1$ of the algorithm. Add $t'-1$ to the set of \emph{barrier levels}
\end{algorithmic}
\end{algorithm}

Our proof proceeds almost in an identical manner to those of Lemmas~\ref{lem:rmfcfeasible} and~\ref{lemma:rmfc2kc}, but now our tree has an additional property that for any two nodes $u \in L_i$ and $v \in L_{i'}$ where $u$ is an ancestor of $v$, the distance between the corresponding points in the metric space $p = \psi_i(u)$ and $q = \psi_{i'}(v)$ is at most $d(p,q) \leq 8 r_i$, which was the property not true in the earlier reduction. This is easy to see because as we traverse a tree path from $u$ to $v$, notice that each time we change winners, the distance between the corresponding points in the metric space decreases geometrically. This proves property~(i) of~\Cref{thm:general-redux}. With this in hand, the remaining proofs to prove the second property are almost identical to the ones in~\Cref{sec:3} and we sketch them below for completeness.

\begin{lemma}
The solution $y$ is a feasible solution to~\ref{lp:rmfc} on $\Itree$ with dilation $1$.
\end{lemma}
\begin{proof}
The proof is via two claims for the two different set of inequalities.
\begin{claim}
For all $1 \leq t \leq h$, we have $\sum_{w \in L_t}{y_w} \leq k_t$.
\end{claim}
\begin{proof} Fix a barrier level $t$. Let $W_t \subseteq X$ denote the winners at level $t$, that is, $W_t = \psi_t(L_t)$.
	By definition of the algorithm,
	$\sum_{w \in L_t} y_w = \sum_{p\in W_t} \cov_t(p)$. Now note that for any two points $p,q \in W_t$, we have $B(p,r_t) \cap B(q,r_t) = \emptyset$.
	To see this, consider the first point which enters $A$ in the round (corresponding to the previous barrier) when $L_t$ was being formed. If this is $p$, then all points in the radius $2r_t$ ball is deleted from $A$. Since the balls are disjoint, the second inequality of \nukc LP implies $\sum_{p\in W_t} \sum_{q \in B_t(p)} x_{q,t} \leq k_t$.
    The second summand in the LHS is the definition of $\cov_t(p)$.
    The same argument holds for all levels $t$ between two consecutive barrier levels $t_1$ and $t_2$ s.t. $t_1 > t_2$, as the winner set remains the same, and the radius $r_t$ is only smaller than the radius $r_{t_2}$ at the barrier $t_2$.
\end{proof}

\begin{claim}
For any leaf node $w \in L$, we have $\sum_{v \in P_w} y_v \geq 1$.
\end{claim}
\begin{proof}
This proof is identical to that of Claim~\ref{cl:coverage}, and we repeat it for completeness. Fix a level $t$ and a winner point $p\in W_t$. Let $u\in L_t$ such that $\psi_t(u) = p$.
	Since $W_t \subseteq W_{t+1} \subseteq \cdots \subseteq W_h$, there is a leaf $v$ in the subtree rooted at $u$ corresponding to $p$.
	Moreover, by the way we formed our tree edges in step (d), we have that $\psi_s(w) = p$ for all $w'$ in the $(u,v)$-path and hence
	$\sum_{w' \in [u,v]\textrm{-path}} y_{w'}$ is precisely $\cov_{\geq t}(p)$.
%

Now, for contradiction, suppose there is some leaf corresponding to, say point $p$, such that the root-leaf path has total $y$-assignment less than $1$. Then, pick the point, among all such unsatisfied points $p$, who appears in a winning set $W_t$ with $t$ \emph{as small as possible}.

By the preceding observation, the total $y$-assignment $p$ receives on its path from level $h$ to level $t$ is exactly $\cov_{\geq t}(p)$.
Moreover, suppose $p$ loses to $q$ at level $t-1$, i.e., $\psi^{-1}_t(p)$ is a child of $\psi^{-1}_{t-1}(q)$. In particular, it means that $q$ has also been a winner up to level $t$ and so the total $y$-assignment on $q$'s path up to level $t$ is also precisely $\cov_{\geq t}(q)$. Additionally, since $\psi^{-1}_{t-1}(q)$ became the parent node for $\psi^{-1}_t(p)$, we know that $\cov_{\geq t}(q) \leq \cov_{\geq t}(p)$ due to the way we choose winners in step (a) of the while loop. Finally, by our maximality assumption on $p$, we know that $q$ is fractionally satisfied by the $y$-solution. Therefore, there is fractional assignment of at least $(1 - \cov_{\geq t}(q))$ on $q$'s path from nodes in level $t-1$ to level $1$. Putting these observations together, we get that the total fractional assignment on $p$'s root-leaf path is at least $\cov_{\geq t}(p) + (1 - \cov_{\geq t}(q)) \geq 1$, which results in the desired contradiction.
\end{proof}
\end{proof}

Finally, the following lemma shows that any good integral solution to the \rmfc instance $\Itree$ can be converted to a good integral solution for the \nukc instance $\I$.

\begin{lemma}
Suppose there exists a feasible solution $N$ to $\Itree$ such that for all $1\leq t\leq h$, $|N\cap L_t| \leq \alpha k_t$.
Then there is a solution to the \nukc instance $\I$ that opens, for each $1\leq t\leq h$,  at most $\alpha k_t$ balls of radius at most $8r_t$.
\end{lemma}

\begin{proof}
Construct the \nukc solution as follows: for level $1 \leq t \leq h$ and every vertex $w \in N\cap L_t$, place the center at $\psi_t(w)$ of radius $8 \cdot r_t$. We claim that every point in $X$ is covered by some ball. Indeed, for any $p\in X$, look at the leaf $v = \psi_{h+1}(p)$, and let $w\in N$ be a node in the root-leaf path which covers it in the instance $\Itree$. By property~(i) of~\Cref{thm:general-redux}, we have that the distance between $\psi_t(w)$ and $p$ is at most $8 \cdot r_t$, and hence the ball of radius $8 \cdot r_t$ around $\psi_t(w)$ covers $p$. The number of balls of radius type $t$ is trivially at most $\alpha k_t$.
\end{proof}

\newcommand{\round}{\ensuremath{{\sf round}}\xspace}
\newcommand{\level}{\ensuremath{{\sf minLevel}}\xspace}
\newcommand{\rad}{\ensuremath{{\widehat{r}}}\xspace}

\section{Getting an $(O(1),O(1))$-approximation algorithm} \label{sec:4}
In this section, we improve our approximation factor on the number of clusters from $O(\log^* n)$ to $O(1)$, while maintaining a constant-approximation in the radius dilation.
As mentioned in the introduction, this requires more ideas since using~\ref{lp:nukc} one cannot get any factor better than $(O(\log^* n),O(1))$-bicriteria approximation since any integrality gap
for \ref{lp:rmfc} translates to a $(\Omega(\log^* n),\Omega(1))$ integrality gap for~\ref{lp:nukc}.

Our algorithm is heavily inspired by the recent paper of Adjiashvili et al~\cite{ABZ16} who give an $O(1)$-approximation for the \rmfc problem. In fact, the structure of our algorithms follows the same three ``steps'' of their algorithm. Given an \rmfc instance,~\cite{ABZ16} first ``compress'' the input tree to get a new tree whose depth is bounded; secondly,~\cite{ABZ16} give a partial rounding result which saves ``bottom heavy'' leaves, that is, leaves which in the LP solution are covered by nodes from low levels; and finally, Adjiashvili et al~\cite{ABZ16} give a clever partial enumeration algorithm for guessing the nodes from the top levels chosen by the optimum solution. We also proceed in these three steps with the first two being very similar to the first two steps in~\cite{ABZ16}. However, the enumeration step requires new ideas for our problem. In particular, the enumeration procedure in~\cite{ABZ16} crucially uses the tree structure of the firefighter instance, and the way our reduction generates the tree for the \rmfc instance is by using the optimal LP solution for the \nukc instance, which in itself suffers from the $\Omega(\log^* n)$ integrality gap. Therefore, we need to devise a more sophisticated enumeration scheme although the basic ideas are guided by those in~\cite{ABZ16}.
%
%
Throughout this section, we do not optimize for the constants.
%

\subsection{Part I: Radii Reduction} \label{sec:part1}

In this part, we describe a preprocessing step which decreases the number of types of radii. This is similar to Theorem 5 in~\cite{ABZ16}.

\begin{theorem}
\label{thm:compress}
	Let $\I$ be an instance of \nukc with radii $\{r_1, r_2,\cdots,r_k\}$. Then we can efficiently construct a new instance $\widehat{\I}$ with radii multiplicities
	$(k_0,\widehat{r}_0),...,(k_L,\widehat{r}_{L})$
	and $L = \Theta(\log k)$ such that:
\begin{enumerate}
\item [(i)] $k_i := 2^i$ for all $0\leq i< L$ and $k_L \le 2^L$. 
\item [(ii)] If the \nukc instance $\I$ has a feasible solution, then there exists a feasible solution for $\widehat{\I}$.
\item [(iii)] Given an $(\alpha,\beta)$-bicriteria solution to $\widehat{\I}$, we can efficiently obtain a $(3\alpha,\beta)$-bicriteria solution to $\I$.
\end{enumerate}
\end{theorem}

\begin{proof}
For an instance $\I$, we construct the compressed instance $\widehat{\I}$ as follows.
Partition the radii into $\Theta(\log k)$ classes by defining barriers at $\widehat{r}_i = r_{2^i}$ for $0 \leq i \leq \floor{\log k}$. Now to create instance $\widehat{\I}$, we simply round up all the radii $r_j$ for $2^{i} \leq j < 2^{i+1}$ to the value $\widehat{r}_i = r_{2^i}$.
Notice that the multiplicity of $\widehat{r}_i$ is precisely $2^i$  (except maybe for the last bucket, where there might be fewer radii rounded up than the budget allowed).

Property (i) is just by construction of instance. Property (ii) follows from the way we rounded up the radii. Indeed, if the optimal solution for $\I$ opens a ball of radius $r_j$ around a point $p$, then we can open a cluster of radius $\widehat{r}_i$ around $p$, where $i$ is such that $2^i \leq j < 2^{i+1}$. Clearly the number of clusters of radius $\widehat{r}_i$ is at most $2^i$ because OPT uses at most one cluster of each radius $r_j$.

For property (iii), suppose we have a solution $\widehat{S}$ for $\widehat{\I}$ which opens $\alpha 2^i$ clusters of radius $\beta \widehat{r}_i$ for all $0\leq i\leq L$.
Construct a solution $S$ for $\I$ as follows. For each $1 \leq i \leq L$, let $C_i$ denote the set of centers where $\widehat{S}$ opens balls of radius $\beta \widehat{r}_i$.
In the solution $S$, we also open balls at precisely these centers with $2\alpha$ balls of radius $r_j$ for every $2^{i-1} \leq j < 2^i$. Since $|C_i|\leq \alpha\cdot 2^i$, we can open a ball at every
point in $C_i$; furthermore, since $j < 2^i$, we have $r_j \geq \widehat{r}_i$ and so we cover whatever the balls from $\widehat{S}$ covered.


Finally, we also open the $\alpha$ clusters (corresponding to $i=0$) of radius $\beta r_1 = \beta \widehat{r}_0$ at the respective centers $C_0$ where $\widehat{S}$ opens centers of radius $\widehat{r}_0$. Therefore, the total number of clusters of radius type is at most $2 \alpha$ with the exception of $r_1$, which may have $3 \alpha$ clusters.
\end{proof}


\subsection{Part II: Satisfying Bottom Heavy Points} \label{sec:part2}\label{sec:4.2}

One main reason why the above height reduction step is useful, is the following theorem from~\cite{ABZ16} for \rmfc instances on trees; we provide a proof sketch for completeness.

\begin{theorem}[\cite{ABZ16}] \label{thm:abz-rnd}
	Given a tree $T$ of height $h$ and a feasible solution $y$ to (\rmfc LP), we can find a feasible integral solution $N$ to \rmfc such that for all $1\leq t\leq h$,
	$|N\cap L_t| \leq k_t + h$.
\end{theorem}
\begin{proof}
	Let $y$ be a basic feasible solution of (\rmfc LP).
	Call a vertex $v$ of the tree {\em loose} if $y_v > 0$ and the sum of $y$-mass on the vertices from $v$ to the root (inclusive of $v$) is $ < 1$.
	Let $V_L$ be the set of loose vertices of the tree, and let $V_I$ be the set of vertices with $y_v = 1$. Clearly $N = V_L \cup V_I$ is a feasible solution: every leaf-to-root path either contains
	an integral vertex or at least two fractional vertices with the vertex closer to root being loose.
	Next we claim that $|V_L|\leq h$; this proves the theorem since $|N\cap L_t| \leq |V_I\cap L_t| + |V_L| \leq k_t + |V_L|$.
	
	The full proof can be found in Lemma 6,~\cite{ABZ16} -- here is a high level sketch. There are $|L|+h$ inequalities in (\rmfc LP), and so the number of fractional variables
	is at most $|L|+h$. We may assume there are no $y_v = 1$ vertices. 
	Now, in every leaf-to-root path there must be at least $2$ fractional vertices, and the one closest to the leaf must be non-loose.
    If the closest fractional vertex to each leaf was unique, then that would account for $|L|$ fractional non-loose vertices implying the number of loose vertices must be $\leq h$.
	This may not be true; however, if we look at {\em linearly independent} set of inequalities that are tight, we can argue uniqueness as a clash can be used to exhibit linear dependence
	between the tight constraints. 
\end{proof}

\begin{theorem} \label{cor:abz-rnd}
Suppose we are given an \nukc instance $\widehat{\cI}$ with radii multiplicities \\
$(k_0,\rad_0), (k_1,\rad_1), \ldots, (k_L,\rad_L)$ with budgets $k_i = 2^i$ for radius type $\rad_i$,
and an LP solution $x$ to {\em (\nukc LP)} for $\widehat{\cI}$.
Let $\tau = \log \log L$, and suppose $X' \subseteq X$ be the points covered mostly by small radii, that is,  let $\cov_{\geq \tau}(p) \geq \frac12$ for every $p \in X'$.
Then, there is an efficient procedure $\round$ which opens at most $O(k_t)$ balls of radius $O(\rad_t)$ for $\tau\leq t\leq L$, and covers all of $X'$.
\end{theorem}

\begin{proof}
The procedure \round works as follows: we partition the points of $X'$ into two sets, one set $X_U$ in which the points receive at least $\frac14$ of the coverage by clusters of radius $\rad_i$ where $i \in \{\log \log L,  \log\log L+1, \ldots, \log L\}$, and another set $X_B$ in which the points receive $\frac14$ coverage from clusters of levels $t \in \{\log L+1, \log L + 2, \ldots, L\}$.
More precisely, $X_U := \{p\in X': \sum_{t=\tau}^{\log L} \cov_t(p) \geq 1/4\}$, and $X_B = X'\setminus X_B$.

Now consider the following LP-solution to (\nukc LP) for $\widehat{\I}$ restricted to $X_U$: we scale $x$ by a factor $4$ and zero-out $x$ on radii type $\rad_i$ for
$i\notin \{\log \log L, \ldots, \log L\}$. By definition of $X_U$ this is a feasible fractional solution; furthermore, the LP-reduction algorithm described in Section~\ref{sec:3} will
lead to a tree $T$ of height $\leq \log L$ and fractional solution $y$ for (\rmfc LP) on $T$ were each $k_i \geq 2^{\log\log L} = \log L$.
Applying ~\Cref{thm:abz-rnd}, we can find an integral solution $N$ with at most $O(k_i)$ vertices at levels $i \in \{\log \log L, \ldots, \log L\}$. We can then translate this solution back using~\Cref{thm:general-redux} to \nukc and find $O(k_t)$ clusters of radius $O(\rad_t)$ to cover all the points $X_U$.
A similar argument, when applied to the smaller radius types $\rad_t$ for $t \in \{\log L, \ldots, L\}$ can cover the points in $X_B$.
\end{proof}

We now show how we can immediately also get a (very weakly) quasi-polynomial time $O(1)$-approximation for \nukc. Indeed, if we could enumerate the set of clusters of radii $\widehat{r}_t$ for $0 \leq t < \log \log L$, we can then explicitly solve an LP where all the uncovered points need to be fractionally covered by only clusters of radius type $\widehat{r}_t$ for $t \geq \log \log L$. We can then round this solution using Corollary~\ref{cor:abz-rnd} to obtain the desired $O(1)$-approximation for the \nukc instance. Moreover, the time complexity of enumerating the optimal clusters of radii $\widehat{r}_t$ for $0 \leq t < \log \log L$ is $n^{O(\log L)} = n^{O(\log \log k)}$, since the number of clusters of radius at least $\widehat{r}_{\log \log L}$ is at most $O(2^{\log \log L}) = O(\log L)$. Finally, there was nothing special in the proof of~\Cref{cor:abz-rnd} about the choice of $\tau = \log \log L$ --- we could set $t = \log^{(q)} L$ to be the $q^{\text th}$ iterated logarithm of $L$, and obtain an $O(q)$-approximation. As a result, we get the following corollary. Note that this gives an alternate way to prove Theorem~\ref{thm:logstar}.

\begin{corollary} \label{cor:qapx}
For any $q \geq 1$, there exists an $(O(q), O(1))$-factor bicriteria algorithm for \nukc which runs in  $n^{O(\log^{(q)} k)}$ time.
\end{corollary}

\subsection{Part III: Clever Enumeration of Large Radii Clusters} \label{sec:part3}
\def\OPT{\mathsf{OPT}}
In this section, we show how to obtain the $(O(1),O(1))$-factor bi-criteria algorithm.
At a high level, our algorithm tries to ``guess'' the centers\footnote{Actually, we end up guessing centers ``close'' to the optimum centers, but for this introductory paragraph this intuition is adequate.}$A$ of large radius, that is $\rad_i$ for $i \leq \tau := \log\log L = \log\log\log k$, which the optimum solution uses.
However, this guessing is done in a cleverer way than in Corollary~\ref{cor:qapx}.
In particular,
given a guess which is consistent with the optimum solution (the initial ``null set'' guess is trivially consistent), our enumeration procedure generates a list of candidate additions to $A$ of size {\em at most} $2^\tau\approx \polyloglog k$ (instead of $n$), one of which is a consistent enhancement of the guessed set $A$. This reduction in number of candidates also requires us to maintain a guess $D$ of points where the optimum solution {\em doesn't} open centers. Furthermore, we need to argue that the ``depth of recursion'' is also bounded by $\polyloglog k$; this crucially uses the technology developed in Section 3.
Altogether, we get the total time is at most $(\polyloglog k)^{\polyloglog k} = o(k)$ for large $k$.

We start with some definitions. Throughout, $A$ and $D$ represent sets of tuples of the form $(p,t)$ where $p\in X$ and $t \in \{0,1,\ldots,\tau\}$.
Given such a set $A$, we associate a partial solution $S_A$ which opens a ball of radius $22\rad_t$ at the point $p$.
For the sake of analysis, fix an optimum solution $\OPT$. We say the set $A$ is {\bf consistent} with $\OPT$ if for all $(p,t)\in A$, there exists a {\em unique} $q\in X$ such that $\OPT$ opens a ball of radius $\rad_t$ at $q$ and $d(p,q)\leq 11\rad_t$. In particular, this implies that $S_A$ covers all points which this $\OPT$-ball covers.
We say the set $D$ is {\bf consistent} with OPT if for all $(q,t)\in D$, OPT {\em doesn't} open a radius $\rad_t$ ball at $q$ (it may open a different radius ball at $q$ though).
Given a pair of sets ($A$,$D$), we define the $\level$ of each point $p$ as follows
\[
\level_{A,D}(p) := 1 + \arg\max_t \{(q,t) \in D ~~\text{for all}~~ q \in B(p,\rad_t)\}
\]
If $(A,D)$ is a consistent pair and $\level_{A,D}(p) = t$, then this implies in the $\OPT$ solution, $p$ is covered by a ball of radius $\rad_t$ or smaller.

%
%
Next, we describe a nuanced LP-relaxation for \nukc.
Fix a pair of sets $(A,D)$ as described above. Let $X_G$ be the subset of points in $X$ covered by the partial solution $S_A$.
Fix a subset $Y\subseteq X\setminus X_G$ of points. Define the following LP.
%
\begin{align}
	\forall p\in Y, & ~~~~~~~~~\sum_{t=\level(p)}^L  ~~\sum_{q \in B(p,\rad_t)} x_{q,t} ~~~\geq 1 \tag{$\text{LP}_\nukc(Y,A,D)$} \label{lp(y,d):nukc} \\
	\forall t\in 1,\cdots,h & ~~~~~~~~~~~~~ \sum_{q\in Y} x_{q,t} ~~~\leq ~~~ k_t \nonumber \\
	\forall (p,t) \in A, & ~~~~~~~~~~~~~~~~ x_{p,t} = 1 \nonumber
\end{align}
\noindent
The following claim encapsulates the utility of the above relaxation.
\begin{claim}\label{clm:simple}
	If $(A,D)$ is consistent with $\OPT$, then~\eqref{lp(y,d):nukc} is feasible.
\end{claim}
\begin{proof}
	We describe a feasible solution to the above LP using $\OPT$. Given $\OPT$, define $O$ to be the collection of pairs $(q,t)$ where
	$\OPT$ opens a radius $\rad_t$ ball at point $q$. Note that the number of tuples in $O$ with second term $t$ is $\leq k_t$.

	Since $A$ is consistent with $\OPT$, for every $(p,t)\in A$, there exists a unique $(q,t) \in O$; remove all such tuples from $O$.
	Define $x_{q,t} = 1$ for all other remaining tuples. By the uniqueness property, we see the second inequality of the LP is satisfied.
	We say a  point $p\in X$ is covered by $(q,t)$ if $p$ lies in the $\rad_t$-radius ball around $q$. Since $Y\subseteq X\setminus X_G$, and since the partial solution $S_A$
	contains all points $p$ which is covered by all the removed $(q,t)$ tuples, we see that every point $p\in Y$ is covered by some remaining $(q,t)\in O$.
	Since $D$ is consistent with $\OPT$, for every point $p\in Y$ and $t < \level_{A,D}(p)$, if $q\in B(p,\rad_t)$ then $(q,t)\notin O$. Therefore, the first inequality is also satisfied.	
 \end{proof}

Finally, for convenience, we define a {\bf forbidden set} $F := \{(p,i): p\in X, 1\leq i\leq \tau\}$ which if added to $D$ disallows any large radii balls to be placed anywhere.

Now we are ready to describe the enumeration~\Cref{alg:enum}. We start with $A$ and $D$ being null, and thus vacuously consistent with $\OPT$.
The enumeration procedure ensures that: given a consistent $(A,D)$ tuple, either it finds a good solution using LP rounding (Step 10), or
generates candidate additions (Steps 18--20) to $A$ or $D$ ensuring that one of them leads to a larger consistent tuple.

\begin{algorithm}
\caption{${\sf Enum}(A,D,\gamma)$}
\label{alg:enum}
\begin{algorithmic}[1]

\STATE Let $X_G = \{ p \, : \, \exists \, (q,i) \in A \text{ s.t } d(p,q) \leq 22 \rad_i \}$ denote points covered by $S_A$.
\IF{there is {\bf no}  feasible solution to $LP_{\nukc}(X \setminus X_G,A,D)$}
	\STATE {\bf Abort.} \\ //{\em Claim \ref{clm:simple} implies $(A,D)$ is not consistent.}
\ELSE
\STATE $x^*$ be a feasible solution to $LP_{\nukc}(X \setminus X_G,A,D)$.
\ENDIF
\STATE Let $X_B = \{ u \in X \setminus X_G : \cov_{\geq \tau}(u) \geq \frac12 \}$ denote bottom-heavy points in $x^*$
\STATE Let $S_B$ be the solution implied by~\Cref{cor:abz-rnd}. \\ // {\em This solution opens $O(k_t)$ balls of radius $O(\rad_t)$ for $\tau\leq t\leq L$ and covers all of $X_B$.}
\STATE Let $X_T = X \setminus (X_G \cup X_B)$ denote the top heavy points in $x^*$
\IF{$LP_\nukc(X_T,A,F\cup D)$ has a feasible solution $x_T$}\label{algo:if}
\STATE By definition of $F$, in $x_T$ we have $\cov_{\geq \tau}(u) = 1$ for all $u\in X_T$.
\STATE $S_T$ be the solution implied by~\Cref{cor:abz-rnd}. \\ // {\em This solution opens $O(k_t)$ balls of radius $O(\rad_t)$ for $\tau\leq t\leq L$ and covers all of $X_T$.}
\STATE Output $(S_A \cup S_B \cup S_T)$. {\em //This is a $(O(1),O(1))$-approximation for the \nukc instance.}\label{algo:op}
\ELSE
\FOR{every level $0 \leq t \leq \tau$}
\STATE Let $C_t=\{p \in X_T$ s.t $\level_{A,D}(p)=t\} $, the set of points in $X_T$ with $\level$ $t$.
\STATE Use the LP-aware reduction from Section~\ref{sec:3} using $x^*$ and the set of points $C_t$ to create tree $T_t$. \label{step:22}
\FOR{every winner $p$ at level $t$ in $T_t$}
\STATE ${\sf Enum}(A\cup\{(p,t)\},D,\gamma -1)$ \label{state:1}
\STATE ${\sf Enum}(A,D\cup \bigcup_{p'\in B(p,11\rad_t)} \{(p', t)\}) , \gamma -1)$ \label{state:2}
\ENDFOR
\ENDFOR
\ENDIF
\end{algorithmic}
\end{algorithm}
Define $\gamma_0 := 4\log\log k\cdot \log\log\log k$.
The algorithm is run with ${\sf Enum}(\emptyset,\emptyset,\gamma_0)$.
The proof that we get a polynomial time $(O(1),O(1))$-bicriteria approximation algorithm follows from three lemmas.
Lemma~\ref{lem:final} shows that if Step~\ref{algo:if} is true with a consistent pair $(A,D)$, then the output in Step~\ref{algo:op} is a $(O(1),O(1))$-approximation.
Lemma~\ref{lem:if} shows that indeed Step~\ref{algo:if} is true for $\gamma_0$ as set. Finally, Lemma~\ref{label:runtime} shows with such a $\gamma_0$, the algorithm runs in polynomial time.
\begin{lemma}\label{lem:final}
	If $(A,D)$ is a consistent pair such that Step~\ref{algo:if} is true, then the solution returned is an $(O(1),O(1))$-approximation algorithm.
\end{lemma}
\begin{proof}
	Since $A$ is consistent with $\OPT$, $S_A$ opens at most $k_t$ centers with radius $\leq 22\rad_t$ for all $0\leq t\leq \tau$.
	By design, $S_B$ and $S_T$ open at most $O(k_t)$ centers with radius $\leq O(r_t)$ for $\tau\leq t\leq L$.
\end{proof}
\begin{lemma} \label{lem:if}
	${\sf Enum}(\emptyset,\emptyset,\gamma_0)$ finds consistent $(A,D)$ such that Step \ref{algo:if} is true.
\end{lemma}
\begin{proof}
	For this we identify a particular execution path of the procedure ${\sf Enum}(A,D,\gamma)$, that at every point maintains a tuple $(A,D)$ that is consistent with OPT. At the beginning of the algorithm, $A=\emptyset \text{ and } D=\emptyset$, which is consistent with $\OPT$.
	
	Now consider a tuple $(A,D)$ that is consistent with $\OPT$ and let us assume that we are within the execution path ${\sf Enum}(A,D,\gamma)$. Let $X\setminus X_G$ be the points not covered by $A$ and let $x^*$ be a solution to $LP_{\nukc}(X\setminus X_G,A,D)$. If $\OPT$ covers all top-heavy points $X_T$ using only smaller radii, then this implies $LP_\nukc(X_T,A,F\cup D)$ has a feasible solution implyin Step \ref{algo:if} is true. So, we may assume, there exists at least one top-heavy point $q\in X_T$ that $\OPT$ covers using a ball radii $\geq \rad_\tau$ around a center $o_q$. In particular, $\level_{A,D}(q) \leq \tau$.
	Let $q \in C_t$ and hence $q$ belongs to $T_t$ for some $0\leq t\leq \tau$. Let $p \in P_t$ be the level $t$ winner in $T_t$ s.t $q$ belongs to the sub-tree rooted at $p$ in $T_i$; $p$ may or may not be $q$.
	We now show that there is at least one recursive call where we make non-trivial progress in $(A,D)$. Indeed, we do this in two cases:
	
	\medskip \noindent {\bf case (A)} If $\OPT$ opens a ball of radius $\rad_t$ at a point $o$ such that $d(o,p)\leq 11\rad_{t}$. In this case, Step~\ref{state:1} maintains consistency.
	Furthermore, we can ``charge'' $(p,t)$ uniquely to the point $o$ with radius $\rad_t$.
	To see this, for contradiction, let us assume that before arriving to the recursive call where $(p,t)$ is added to $A$, some other tuple $(u,t) \in A^{'}$, in an earlier recursive call with $(A^{'},D^{'})$ as parameters charged to $(o,t)$. Then by definition we know that $d(u,o) \leq 11\rad_t$ implying $d(u,p) \leq 22\rad_t$. Then $p$ would be in $X_G$ in all subsequent iterations, contradicting that $p\in X_T$ currently.
	
	\medskip \noindent {\bf case (B)}: Tf there is no $(o,t) \in \OPT$ with $d(o,p)\leq 11\rad_{t}$, then for all points $p'\in B(p,11\rad_t)$ we can add $(p',t)$ to $D$.
	In this case, we follow the recursive call in Step~\ref{state:2}.
	
	To sum, we can definitely follow the recursive calls in the consistent direction; but how do we bound the depth of recursion.
	In case (A), the measure of progress is clear -- we increase the size of $|A|$, and it can be argued (we do so below) the maximum size of $A$ is at most $\polyloglog k$.
	Case (B) is subtler. We do increase size of $D$, but $D$ could grow as large as $\Theta(n)$. Before going to the formal proof, let us intuitively argue what ``we learn'' in Case (B).
	Recall $q$ is covered in $\OPT$ by a ball around the center $o_q$.
	Since $\level(q) = t \leq \tau$, by definition there is a point $v\in B(q,\rad_t)$ such that $(v,t)\notin D$, and $d(q,o_q) \leq \rad_t$.
	Together, we get $d(v,o_q) \leq 2\rad_t$, that is, $v\in B(o_q,2\rad_t)$.
	Now also note, since $q$ lies in $p$'s subtree in $T_t$, by construction of the trees, $d(p,q) \leq 8\rad_t$ by property~(i) of~\Cref{thm:general-redux}. Therefore, $d(p,o_q) \leq 9\rad_t$ and in case (B), for {\em all} points $u\in B(o_q,2\rad_t)$ we put $(u,t)$ in the set $D$ in the next recursive call.
	This is ``new info'' since for the current $D$ we know that at least one point $v\in B(o_q,2\rad_t)$, we had $(v,t)\notin D$.
	
	Formally, we define the following potential function. Let $O_\tau$ denote the centers in $\OPT$ around which balls of radius $\rad_j$, $j\leq \tau$ have been opened.
	Given the set $D$, for $0\leq t\leq \tau$ and for all $o\in O_\tau$, define the indicator variable $Z^{(D)}_{o,t}$ which is $1$ if for all points $u\in B(o,2\rad_t)$, we have $(u,t)\in D$ and $0$ otherwise.
	\[
	\Phi(A,D) :=  |A| ~~ + ~~ \sum_{o \in O_\tau} \sum^{\tau}_{t=0} Z^{(D)}_{o,t}
	\]
	\noindent
	Note that $\Phi(\emptyset,\emptyset) = 0$. From the previous paragraph, we conclude that in both case (A) or (B), the potential increases by {\em at least} $1$. Finally, for any consistent $A,D$ we can upper bound $\Phi(A,D)$ as follows. Since $A$ is consistent, $|A| \leq \sum_{t=0}^\tau 2^t \leq 2^{\tau+1} = \log L = \log\log k$. The second term in $\Phi$ is at most
	$2^{\tau + 1}\cdot \tau = \log L\log\log L$.  Thus, in at most $2\log\log k\cdot \log\log\log k < \gamma_0$ steps we reach a consistent pair $(A,D)$ with Step~\ref{algo:if} true.
\end{proof}

%
%

\begin{lemma}
\label{label:runtime}
${\sf Enum}(\emptyset,\emptyset,\gamma_0)$ runs in polynomial time for large enough $k$.
\end{lemma}
\begin{proof}
	Each single call of ${\sf Enum}$ is clearly polynomial time, and so we bound the number of recursive calls. Indeed, this number will be $o(k)$.
	We  first bound the number of recursive calls in a single execution of Enum$(A,D,\gamma)$.
	For a fixed tuple $(A,D)$, \Cref{alg:enum}, constructs trees $T_0,\ldots,T_{\tau}$ in Step \ref{step:22}, using the reduction algorithm from Section~\ref{sec:3}.
	Let $L_{tj}$ represent the set of nodes at level $j$ in the tree $T_t$. Then $P_t = \psi_{t}(L_{tt})$ represents the set of points that are winners at level $t$ in $T_t$.
    Now for any tree $T_t$, \Cref{alg:enum} makes two recursive calls for each winner in $P_t$ (Step \ref{state:1}-\ref{state:2}). Let $P_{AD}=\bigcup_{t=0}^{\tau}{P_t}$ be the set of all the winners that the algorithm considers in a single call to ${\sf Enum}(A,D,\gamma)$. The total number of recursive calls in a single execution is therefore $2|P_{AD}|$. Now we claim that for a fixed tuple $(A,D)$, the total number of winners is bounded.
	
	\begin{claim}
		$|P_{AD}| \leq 4\log \log k \cdot \log \log \log k$
	\end{claim}
	
	\begin{proof}
		Consider a tree $T_t$ and the corresponding set $P_t$ as defined above. We use $y^{(t)}$ to denote the \rmfc LP solution given along with the tree $T_t$ given by the reduction algorithm in Section 3. Let $p \in P_t$ be a winner at level $t$ (and consequently, at level $\tau$ also), and suppose it is mapped to tree vertices $w$ at level $t$ and $w'$ at level $\tau$. Then, by the way the tree was constructed and because $p\in X_T$ is top-heavy, we have $\sum_{u \in [w,w']-\text{path}} y^{(t)}_u \geq \frac12$ (Refer to the proof of Claim~\ref{cl:coverage} for more clarity).
		So each winner at level $t$ has a path down to level $\tau$ with fractional coverage at least $\frac12$. But the total fractional coverage in the top part of the tree is at most the total budget, which is $\sum_{t=1}^{\tau} 2^{t} \leq 2 \log L \leq 2 \log \log k$. Therefore, $|P_t| \leq 4 \log \log k$. Adding for all $1\leq t\leq \tau$, gives $|P_{AD}| \leq 4\log\log k \cdot t \leq 4 \log \log k \cdot \log \log \log k$.
	\end{proof}
	Since the recursion depth $\gamma_0$, the total number of recursive calls made to the Enum is loosely upper bounded by $\gamma_0^{\polyloglog k} = o(k)$, thus completing the proof.	
\end{proof}

\noindent

\bibliography{paper}
\appendix

\end{document}